\documentclass[11pt, a4paper]{article}

\usepackage{amsthm}
\usepackage[utf8]{inputenc}
\usepackage{geometry}
\usepackage{amsmath, amssymb}
\usepackage{algorithm}
\usepackage{algpseudocode}
\usepackage{listings}
\usepackage{xcolor}
\usepackage{hyperref}
\usepackage{graphicx}
\usepackage{dirtytalk}
\usepackage{tikz}
\usepackage{enumitem}
\usepackage{subcaption}

\usetikzlibrary{positioning}

\newtheorem{fact}{Fact} 
\newtheorem{theorem}{Theorem}
\newtheorem{lemma}[theorem]{Lemma}

\newtheorem{corollary}{Corollary}

\definecolor{codegreen}{rgb}{0,0.6,0}
\definecolor{codegray}{rgb}{0.5,0.5,0.5}
\definecolor{codepurple}{rgb}{0.58,0,0.82}
\definecolor{backcolour}{rgb}{0.95,0.95,0.92}

\lstdefinestyle{mystyle}{
    backgroundcolor=\color{backcolour},   
    commentstyle=\color{codegreen},
    keywordstyle=\color{magenta},
    numberstyle=\tiny\color{codegray},
    stringstyle=\color{codepurple},
    basicstyle=\ttfamily\footnotesize,
    breakatwhitespace=false,         
    breaklines=true,                 
    captionpos=b,                    
    keepspaces=true,                 
    numbers=left,                    
    numbersep=5pt,                  
    showspaces=false,                
    showstringspaces=false,
    showtabs=false,                  
    tabsize=2,
    language=Python
}

\title{\textbf{Dynamic framework for edge-connectivity maintenance of simple graphs}}
\author{B{\l}a{\. z}ej Wr{\' o}bel \thanks{Wroc{\l}aw University of Science and Technology, Poland}}
\date{\today}

\begin{document}

\maketitle
\abstract{
We present a framework for dynamically maintaining $k$-edge-connectivity of an undirected simple graph $G$ under edge insertions and deletions, where $k$ is a fixed constant. After an edge insertion, the algorithm identifies and removes a distinct redundant edge to maintain sparsity, in $O(k \log n)$ amortized time. After an edge deletion that reduces $\lambda(G)$ below $k$, the algorithm restores $k$-edge-connectivity by adding at most two new edges (excluding the deleted edge), in $O(k^{3/2} n^{3/2})$ time. The insertion procedure combines Nagamochi–Ibaraki sparse certificates with Link-Cut Trees; the deletion procedure uses a single maximum-flow computation on the sparsified graph. Throughout all updates, the graph is maintained with $O(kn)$ edges.
}

\section{Introduction}
\label{section:introduction}
Edge-connectivity is one of the fundamental measures of network reliability and fault tolerance. A graph is $k$-edge-connected if the removal of fewer than $k$ edges cannot disconnect it. We study the problem of actively maintaining $k$-edge-connectivity: given a graph $G$ with $\lambda(G) \geq k$, the goal is to preserve this invariant under edge insertions and deletions by modifying the graph itself -- not merely reporting whether the invariant holds. An edge insertion may make an existing edge redundant for $k$-edge-connectivity, while an edge deletion may reduce $\lambda(G)$ below $k$. We address both cases:

\begin{enumerate}
    \item Redundancy Elimination (Post-Addition): after insertion of a new edge $e_{new}=\{u,v\}$, local edge-connectivity between its endpoints may be greater than $k$. To maintain structural sparsity -- specifically ensuring $|E| = O(kn)$ -- we identify and remove an existing edge $e_{old}=\{x,y\}$ ($e_{old} \neq e_{new}$) for which the condition $\lambda(x, y; (G \cup \{e_{new}\}) \setminus \{e_{old}\}) \ge k$ (see   Section~\ref{section:preliminaries-and-notation} for definition of $\lambda(x, y; G)$) is met. This ensures that the graph remains sparse and $k$-edge-connected.

    \item Connectivity Restoration (Post-Removal): after removal of an existing edge $e_{del}$, the global edge-connectivity may drop below the required threshold, i.e., $\lambda(G \setminus \{e_{del}\}) < k$. In such cases, the framework computes and inserts a minimal set of augmenting edges $E_{aug}$ ($e_{del} \notin E_{aug}$) such that the condition $\lambda((G  \setminus \{e_{del}\}) \cup E_{aug}) \geq k$ holds. 
    
\end{enumerate} \par

Beyond its theoretical interest, our framework is motivated by settings where a network's link structure must be actively maintained under dynamic changes. A natural application of Connectivity Restoration arises in survivable telecommunications network design, where nodes represent routers and edges represent physical communication links. Such networks are designed to be sparse and $k$-edge-connected so that data transmission survives up to $k - 1$ simultaneous link failures~\cite{Grotschel92}. When a link is lost -- due to hardware malfunction or fiber cut -- the operator must provision new links to restore the required redundancy. Our framework models this directly: the failed link is an edge deletion, and the algorithm computes the minimal set of new links to restore invariant $\lambda(G) \geq k$. \par

In wireless sensor networks, topology control protocols seek sparse communication topologies that preserve connectivity, since maintaining fewer active links reduces both energy consumption and radio interference~\cite{Santi2005}. Our Redundancy Elimination task provides an algorithmic primitive aligned with this objective: when a new link becomes available, the framework detects and removes an existing link that is no longer necessary for $k$-edge-connectivity, keeping the network sparse with $O(kn)$ edges. \par

\section{Related work}
\label{section:related-work}

Computing the minimum cut in an undirected graph is a classic problem with a long history. Early approaches relied on network flow techniques; the Gomory-Hu~\cite{Gomory61} tree construction reduces the problem to $n - 1$ maximum flow computations. Karger~\cite{Karger2000} introduced a randomized algorithm based on tree packings that finds the global minimum cut in $O(m \log^{3}n)$ time with high probability. In the deterministic setting, Henzinger, Rao, and Wang~\cite{Henzinger2020} gave an algorithm for determining graph's edge-connectivity, running in $O(m \log^{2}(n)  \log(\log^{2}(n)))$ time. Li~\cite{Li2021} presented a deterministic algorithm for finding minimum-cuts in weighted undirected graphs running in $m^{1+o(1)}$ time. \par

The problem of maintaining the minimum cut in a dynamically changing graph has been studied extensively, with a primary focus on efficiently reporting the value of the minimum cut rather than actively enforcing a connectivity invariant. Thorup~\cite{Thorup2007} gave a deterministic fully dynamic algorithm that maintains the value of the minimum cut in $\tilde{O}(\sqrt{n})$ (see Section~\ref{section:preliminaries-and-notation}) worst-case time per update. For the incremental setting, Goranci et al.~\cite{Goranci2018} achieved $O (\log^{3}(n)\log(\log^{2}(n)))$  amortized update time by maintaining a set of edge-disjoint spanning forests -- a sparse certificate that preserves connectivity properties. Our approach to redundancy elimination (Section~\ref{section:restoration-after-edge-addition}) builds on a similar structural insight: by actively maintaining a certificate of size $O(kn)$, we can efficiently identify and remove redundant edges in $O(k \log n)$ amortized time per edge addition. For connectivity restoration (Section~\ref{section:restoration-after-edge-removal}), we exploit the specific properties of Dinic's algorithm~\cite{Dinitz70} on unit-capacity networks. While the general weighted case requires $O(n^{2}m)$ time, Even and Tarjan~\cite{Shimon75} showed that the complexity improves to $O(\min(n^{2/3}, m^{1/2})m)$ for unit capacities. This bound is critical to our edge-connectivity restoration strategy.

\section{Preliminaries and notation}
\label{section:preliminaries-and-notation}

We consider unweighted, simple, undirected graphs $G = (V, E)$ with a vertex set $V$ and edge set $E$, where $|V| = n$ and $|E| = m$, using these parameters primarily to describe the computational complexity of the proposed algorithms. To simplify the asymptotic analysis, we employ the notation $\tilde{O}(f(n))$ to denote $O(f(n) \log^k n)$ for some constant $k \ge 0$, thereby suppressing polylogarithmic factors; crucially, this notation never obscures polynomial factors. When considering several graphs, we denote the respective sets of a graph $G$ as $V(G)$ and $E(G)$ to avoid ambiguity. Furthermore, we denote an undirected edge connecting vertices $u$ and $v$ as the set $\{u, v\}$, while a directed edge (arc) from $u$ to $v$ is denoted by the ordered pair $(u, v)$. For an edge $e = \{u, v\}$, we write $\text{endpoints}(e)$ to denote the pair of vertices incident to $e$. In pseudocode, the assignment $(u, v) \leftarrow \text{endpoints}(e)$ binds the two endpoints of $e$ to the variables $u$ and $v$ in arbitrary order. \par

For any two vertices $u, v \in V$, let $\lambda(u, v; G)$ denote the local edge-connectivity, defined as the maximum number of edge-disjoint paths between $u$ and $v$ in $G$. The global edge-connectivity of $G$, denoted by $\lambda(G)$, is the minimum local edge-connectivity over all pairs of vertices, i.e. $\lambda(G) = \min\limits_{u, v \in V} \lambda(u, v; G)$. \par 

The parameter $k$ represents the required global edge-connectivity and is treated as a fixed integer constant throughout the paper. \par

For any subset of vertices $S \subseteq V$ and $S \neq \emptyset$, let $G[S] = (S, E[S])$ denote the subgraph induced by $S$, where $E[S] = \{\{u, v\} \in E \mid u, v \in S\}$ represents the set of edges with endpoints entirely contained within $S$. A cut is a partition of the vertex set into two disjoint sets $(S, \bar{S})$ such that $S \cup \bar{S} = V$ and $S \cap \bar{S} = \emptyset$.  The cardinality of the cut is the number of edges with one endpoint in $S$ and the other in $\bar{S}$. Namely,

$$|(S, \bar{S})| = |E \setminus (E[S] \cup E[\bar{S}])|.$$

Throughout this paper, we use the terms cut size and cut cardinality interchangeably. By Menger’s theorem~\cite{Diestel2017}, the local edge-connectivity $\lambda(u, v; G)$ is equivalent to the size of the minimum cut separating $u$ and $v$. \par

We use $G \cup \{e\}$ and $G \setminus \{e\}$ to represent the graph resulting from the insertion or deletion of edge $e$, respectively. In the context of our edge-connectivity maintenance algorithms, $e_{new}$ refers to a newly inserted edge, while $e_{old}$ refers to an existing edge identified for removal during Redundancy Elimination. Similarly, $e_{del}$ denotes an edge targeted for deletion by Connectivity Restoration, and $E_{aug}$ represents the set of augmenting edges computed by the framework to restore edge-connectivity. \par

To address problems defined in Section~\ref{section:introduction} efficiently, we rely on a fundamental graph-theoretic property regarding the locality of connectivity changes. The following lemma characterizes the structural impact of removing an edge from a $k$-edge-connected graph,

\begin{lemma}
\label{lemma:cut-localization-lemma}
Let $G = (V, E)$ be an undirected and $k$-edge-connected graph and $e = \{u, v\}$ be an edge in $E$. Let $G' = G \setminus \{e\}$. If $\lambda(G') < k$, then $\lambda(u, v; G) = k - 1$. Furthermore, any cut in $G'$ which cardinality is less than $k$ must separate $u$ and $v$.
\end{lemma}
\begin{proof}
Consider any cut $(S, \bar{S})$ in $G'$ with cardinality strictly less than $k$. Since the original graph $G$ is $k$-edge-connected, the cardinality of this cut in $G$ must have been at least $k$. The removal of the single edge $e=\{u, v\}$ can reduce size of a cut by at most one, and this occurs if and only if $e$ crosses the cut (i.e., $u \in S$ and $v \in \bar{S}$ or vice versa). Therefore, for the cut's capacity to drop below $k$ in $G'$, its original capacity must have been exactly $k$, and it must separate $u$ and $v$, resulting in a final cardinality of exactly $k-1$.
\end{proof}

We use this lemma directly in the restoration algorithm of Section~\ref{section:restoration-after-edge-removal}.\par

\section{Restoration after Edge Addition}
\label{section:restoration-after-edge-addition}

\subsection{Sparse certificate maintenance}
\label{subsection:sparse-certificate-maintenance}
To identify and eliminate redundant edges, we use a sparse certificate as introduced by Nagamochi and Ibaraki~\cite{Nagamochi92}. The edge set of a graph $G$ can be decomposed into a sequence of edge-disjoint spanning forests $F_{1}, F_{2},\ldots, F_{m}$. For maintaining $k$-edge-connectivity, it suffices to maintain only the first $k$ forests, which form the sparse certificate $G_{cert} = \bigcup_{i=1}^{k} F_i$. The following lemma establishes that this subgraph preserves the relevant edge-connectivity properties:

\begin{lemma}[From~\cite{Nagamochi92}]
\label{lemma:sparse-cert-lemma}
For a simple graph $G = (V, E)$, let $F_i = (V, E_i)$ be a
maximal spanning forest in $G - E_1 \cup E_2 \cup \dots \cup E_{i-1}$ for $i = 1, 2, \dots, |E|$, where
possibly $E_i = E_{i+1} = \dots = E_{|E|} = \emptyset$ for some $i$. Then each spanning subgraph
$G_i = (V, E_1 \cup E_2 \cup \dots \cup E_i)$ satisfies
\begin{equation}
\lambda(x, y; G_i) \ge \min\{\lambda(x, y; G), i\} \quad \text{for all } x, y \in V,
\end{equation}
where $\lambda(x, y; H)$ denotes the local edge-connectivity between $x$ and $y$ in graph $H$.
\end{lemma}

In particular, if $k \leq \lambda(G)$, the subgraph $G_{k} = \left(V, \bigcup_{i = 1}^{k} E_{i}\right)$ is itself $k$-edge-connected. \par

\begin{algorithm}
\caption{Finding $k$-edge-connected spanning subgraph of a given graph $G$ (from~\cite{Nagamochi92}).}
\label{algorithm:finding-k-edge-connected-spanning-subgraph}
\begin{algorithmic}[1]
\State \textbf{Input:} Graph $G = (V, E)$
\State \textbf{Output:} Partition $E_1, E_2, \dots, E_{|E|}$ of $E$
\State $E_1 := E_2 := \dots := E_{|E|} := \emptyset$;
\State Label all nodes $v \in V$ and all edges $e \in E$ ``unscanned'';
\State $r(v) := 0$ for all $v \in V$;
\While{there exist ``unscanned'' nodes}
    \State Choose an ``unscanned'' node $x \in V$ with the largest $r(x)$; \label{line:choosing-with-largest}
    \For{each ``unscanned'' edge $e$ incident to $x$}
        \State Let $y$ be the other end node ($\neq x$) of $e$;
        \State $E_{r(y)+1} := E_{r(y)+1} \cup \{e\}$; \label{line:edge-addition}
        \If{$r(x) = r(y)$}
            \State $r(x) := r(x) + 1$;
        \EndIf
        \State $r(y) := r(y) + 1$;
        \State Mark $e$ ``scanned'';
    \EndFor
    \State Mark $x$ ``scanned'';
\EndWhile
\end{algorithmic}
\end{algorithm}

The sparse certificate is constructed using the linear-time sparsification method of Nagamochi and Ibaraki~\cite{Nagamochi92}, presented in Algorithm~\ref{algorithm:finding-k-edge-connected-spanning-subgraph}. The procedure partitions $E$ into a sequence of edge-disjoint forests $E_{1}, E_{2},\ldots, E_{|E|}$ in a single pass over the graph. Each vertex $v$ maintains a rank $r(v)$, initially set to zero, that tracks the number of forests containing an edge incident to $v$. At each step, the algorithm selects an unscanned vertex $x$ with maximum rank and assigns each unscanned edge $e = \{x, y\}$ to forest $E_{r(y)+1}$. Since $y$ has no edge in $E_{r(y)+1}$ at the time of processing, this assignment preserves acyclicity. Ranks are then updated: $r(x)$ (if $r(x) = r(y)$) and $r(y)$ are incremented. Each edge is thus placed into the lowest-index forest that can accept it in $O(1)$ time, yielding an overall running time of $O(m)$. We use this procedure as a preprocessing step to construct the sparse certificate before handling dynamic edge additions. \par

For dynamic maintenance of the sparse certificate we use the Link-Cut tree data structure of Sleator and Tarjan~\cite{Tarjan81}. A Link-Cut tree maintains a forest of rooted trees. Each represented tree is partitioned into vertex-disjoint paths via a preferred-path decomposition, and each such path is stored as a splay tree~\cite{Tarjan85} keyed by the depth of its nodes in the represented tree. The auxiliary splay trees are connected by path-parent pointers, which link the root of each auxiliary tree to the parent (in the represented tree) of the topmost node on its preferred path. We use the following operations~\cite{Tarjan81,Demaine2012}:
\begin{itemize}
    \item \textsc{MakeTree}($v$) -- creates a new singleton tree containing the vertex $v$,
    \item \textsc{FindRoot}($v$) -- returns the root of the tree containing vertex $v$,
    \item \textsc{link}($v$, $w$) --  adds an edge making $v$ a child of $w$, assuming $v$ is the root of its tree and $v$, $w$ belong to distinct trees,
    \item \textsc{cut}($v$) -- removes the edge between vertex $v$ and its parent \textsc{parent}($v$) (where $v$ is not the root), splitting the tree into two disjoint components and makes $v$ the root of a represented tree,
    \item \textsc{evert}($v$) -- modify the tree containing vertex $v$ by making $v$ its root.
\end{itemize}  

Each operation runs in $O(\log n)$ amortized time. To initialize the structure, we compute the sparse certificate using Algorithm~\ref{algorithm:finding-k-edge-connected-spanning-subgraph} and construct $k$ Link-Cut trees $LC_{1}, \dots, LC_{k}$, where each $LC_{i}$ is built by executing \textsc{MakeTree}($v$) for every $v \in V$ followed by \textsc{link} operations for every edge in the corresponding forest $F_{i}$. \par

\begin{algorithm}
\caption{Function to modify forest represented by Link-Cut structure.}
\label{algorithm:try-add}
\begin{algorithmic}[1]
\State \textbf{function} \textsc{TryAdd}($LC, u, v$)
    \State \quad $uRoot \leftarrow LC.\textsc{FindRoot}$($u$) \label{line:find-u-root}
    \State \quad $vRoot \leftarrow LC.\textsc{FindRoot}$($v$) \label{line:find-v-root}

    \item[]

    \If{$uRoot \neq vRoot$}
        \State $LC.\textsc{evert}$($u$)
        \State $LC.\textsc{link}$($u$, $v$) \label{line:invocation-of-link}
        \State \Return \textbf{null}
    \Else
        \State $p \leftarrow LC.\textsc{parent}$($u$)
        \State $LC.\textsc{cut}$($u$) \label{line:cut-u}
        \State $LC.\textsc{link}$($u$, $v$) \label{line:link-u-vroot}
        \State \Return $(u, p)$
    \EndIf
\State \textbf{end function}
\end{algorithmic}
\end{algorithm}

\begin{algorithm}[p]
\caption{Handling edge addition in a sparse certificate.}
\label{algorithm:handling-edge-additions}
\begin{algorithmic}[1]
\State \textbf{procedure} \textsc{HandleAddition}($e$)
\State $(u, v) \leftarrow \text{endpoints}(e)$ \label{line:retrieving-edge-enpoints} \Comment{1. Define endpoints of edge e}

    \State $result \leftarrow \textsc{TryAdd}(LC_1, u, v)$
    \item[]
    \If{$result \neq$ \textbf{null}}
        \State $(x, y) \leftarrow \text{endpoints}(result)$
        \State $i \leftarrow 2$
        \While{$i \leq k$} \label{line:while-loop-beg}
            \State $result \leftarrow \textsc{TryAdd}$($LC_{i}$, $x$, $y$)
            \If{$result = \textbf{null}$}
                \textbf{break}
            \Else
                \State $(x, y) \leftarrow \text{endpoints}(result)$ 
                \State $i \leftarrow i + 1$
            \EndIf
        \EndWhile
    \EndIf
\State \textbf{end procedure}
\end{algorithmic}
\end{algorithm}

\begin{figure}[p]
  \centering
  \begin{tikzpicture}
    \def\offsetA{0}
    \def\offsetB{5}
    \def\offsetC{10}

    \begin{scope}[xshift=\offsetA cm]
      \node[font=\bfseries] at (1, 2.2) {$F_1$};
     
      \node[draw, circle, inner sep=2pt, minimum size=20pt] (1a1) at (0, 1) {$v_1$};
      \node[draw, circle, inner sep=2pt, minimum size=20pt] (1a2) at (1, 1) {$v_2$};
      \node[draw, circle, inner sep=2pt, minimum size=20pt] (1a3) at (2, 1) {$v_3$};
      
      \node[draw, circle, inner sep=2pt, minimum size=20pt] (1b1) at (0, -1) {$v_4$};
      \node[draw, circle, inner sep=2pt, minimum size=20pt] (1b2) at (1, -1) {$v_5$};
      \node[draw, circle, inner sep=2pt, minimum size=20pt] (1b3) at (2, -1) {$v_6$};
      
      \draw[thick] (1a1) -- (1a2);         
      \draw[thick] (1a2) -- (1a3);          
      \draw[thick] (1a1) -- (1b2);         
      \draw[thick] (1a1) -- (1b3);         
     
      \draw[orange, dashed, thick] (1a1) -- (1b1);  
     
      \draw[red, dashed, thick] (1b1) to[bend right=60] (1b3);    
    \end{scope}

    \draw[->, >=stealth, thick, orange, dashed] (2.4, 0) -- node[above, font=\footnotesize] {$\{v_1, v_4\}$} (4.7, 0);

    \begin{scope}[xshift=\offsetB cm]
      \node[font=\bfseries] at (1, 2.2) {$F_2$};
      
      \node[draw, circle, inner sep=2pt, minimum size=20pt] (2a1) at (0, 1) {$v_1$};
      \node[draw, circle, inner sep=2pt, minimum size=20pt] (2a2) at (1, 1) {$v_2$};
      \node[draw, circle, inner sep=2pt, minimum size=20pt] (2a3) at (2, 1) {$v_3$};
     
      \node[draw, circle, inner sep=2pt, minimum size=20pt] (2b1) at (0, -1) {$v_4$};
      \node[draw, circle, inner sep=2pt, minimum size=20pt] (2b2) at (1, -1) {$v_5$};
      \node[draw, circle, inner sep=2pt, minimum size=20pt] (2b3) at (2, -1) {$v_6$};
     
      \draw[thick] (2a2) -- (2b1);         
      \draw[thick] (2a3) -- (2b1);          
      \draw[thick] (2a2) -- (2b2);          
      \draw[thick] (2a2) -- (2b3);         
     
      \draw[green!60!black, dashed, thick] (2a1) -- (2b1);  
    \end{scope}

    \begin{scope}[xshift=\offsetC cm]
      \node[font=\bfseries] at (1, 2.2) {$F_3$};
     
      \node[draw, circle, inner sep=2pt, minimum size=20pt] (3a1) at (0, 1) {$v_1$};
      \node[draw, circle, inner sep=2pt, minimum size=20pt] (3a2) at (1, 1) {$v_2$};
      \node[draw, circle, inner sep=2pt, minimum size=20pt] (3a3) at (2, 1) {$v_3$};
     
      \node[draw, circle, inner sep=2pt, minimum size=20pt] (3b1) at (0, -1) {$v_4$};
      \node[draw, circle, inner sep=2pt, minimum size=20pt] (3b2) at (1, -1) {$v_5$};
      \node[draw, circle, inner sep=2pt, minimum size=20pt] (3b3) at (2, -1) {$v_6$};
      
      \draw[thick] (3a3) -- (3b2);         
      \draw[thick] (3b1) -- (3b2);          
      \draw[thick] (3b2) -- (3b3);         
    \end{scope}

    \begin{scope}[yshift=-2.5cm, xshift=3cm]
      \draw[red, dashed, thick] (0, 0) -- (0.8, 0);
      \node[right, font=\footnotesize] at (0.9, 0) {inserted edge};
      \draw[orange, dashed, thick] (3.5, 0) -- (4.3, 0);
      \node[right, font=\footnotesize] at (4.4, 0) {displaced edge};
      \draw[green!60!black, dashed, thick] (7, 0) -- (7.8, 0);
      \node[right, font=\footnotesize] at (7.9, 0) {absorbed edge};
    \end{scope}
  \end{tikzpicture}
  \caption{\textsc{HandleAddition}($\{v_4, v_6\}$): edge $\{v_4, v_6\}$ displaces $\{v_1, v_4\}$ from $F_1$; edge $\{v_1, v_4\}$ is absorbed into $F_2$. $F_3$ remains unchanged.}
  \label{fig:cascade}
\end{figure}

The core building block of the maintenance procedure is the \textsc{TryAdd}($LC$, $u$, $v$) function (Algorithm~\ref{algorithm:try-add}), which attempts to insert an edge $\{u, v\}$ into the forest $F_i$ maintained by $LC_i$ $(1 \leq i \leq k)$. If $u$ and $v$ lie in different components of $F_i$, the edge is added via \textsc{link} and the function returns null. Otherwise, $u$ and $v$ are already connected in $F_i$, so adding $\{u, v\}$ would create a cycle. In this case, the function removes the edge $\{u, \textsc{parent}(u)\}$ via \textsc{cut}, inserts $\{u, v\}$ in its place via \textsc{link}, and returns the displaced edge $\{u, \textsc{parent}(u)\}$ to the caller for insertion into the next forest in the sparse certificate. \par

The primary maintenance routine, \textsc{HandleAddition}($e$) (Algorithm~\ref{algorithm:handling-edge-additions}), uses \textsc{TryAdd} as a subroutine in a cascading fashion. Given a new edge $e = \{u, v\}$, the procedure calls \textsc{TryAdd} on $LC_1$. If the edge is absorbed (\textsc{TryAdd} returns null), the process terminates. Otherwise, the displaced edge is passed to \textsc{TryAdd} on $LC_2$, and so on through the sequence of forests up to $LC_k$. If an edge is displaced from the final forest $F_k$, it is redundant for $k$-edge-connectivity and is discarded. \par

\subsection{Correctness analysis}
\label{subsection:correctness-analysis}

We first establish that \textsc{TryAdd} preserves the forest structure (Fact~\ref{fact:spanning-forest-properties}, Corollary~\ref{corollary:try-add-correctness}), then that component connectivity is invariant under edge swaps (Lemma~\ref{lemma:sets-of-connected-components}), and finally that \textsc{HandleAddition} preserves $k$-edge-connectivity (Lemma~\ref{lemma:two-endpoints-connected-in-fk}, Theorem~\ref{theorem:edge-cascade-correctness}).

\begin{fact}
\label{fact:spanning-forest-properties}
Let $F = (V, E_F)$ be a spanning forest and let $e = \{u, v\} \notin E_F$. 
\begin{enumerate}
    \item If $u$ and $v$ are in different connected components of $F$, then $(V, E_F \cup \{e\})$ is a forest.
    \item If $u$ and $v$ are in the same connected component and $P_{uv}$ is the path connecting them in $F$, then for any $e' \in P_{uv}$, the graph $(V, (E_F \cup \{e\}) \setminus \{e'\})$ is a spanning forest.
\end{enumerate}
\end{fact}

\begin{corollary}
\label{corollary:try-add-correctness}
 Let $F = (V, E_F)$ be the forest represented by Link-Cut tree $LC$, and let $e = \{u, v\}$ be an edge not in $E_F$. After executing \textsc{TryAdd}($LC$, $u$, $v$), the forest represented by $LC$ remains a valid spanning forest.
\end{corollary}

\begin{lemma}
\label{lemma:sets-of-connected-components}
Let $F_i = (V, E_i)$ be a spanning forest and let $F'_i$ be the forest obtained by replacing an edge $\{u, \textsc{parent}(u)\} \in E_i$ with an edge $\{u, v\} \notin E_i$, where $u$ and $v$ are in the same connected component of $F_i$. For any pair of vertices $x, y \in V$, $x$ and $y$ are connected in $F_i$ if and only if they are connected in $F'_i$.
\end{lemma}
\begin{proof}

Let $C$ be the component of $F_i$ containing $u$ and $v$. Removing $\{u, \textsc{parent}(u)\}$ partitions $V(C)$ into two subsets: $V_u$, the vertices of the subtree rooted at $u$, and $V(C) \setminus V_u$. Since $v \in V(C) \setminus V_u$, adding $\{u, v\}$ reconnects these two subsets into a single connected component. No edges incident to vertices outside $V(C)$ are affected, so the remaining components of $F_i$ are unchanged.
\end{proof}

\begin{lemma}
\label{lemma:two-endpoints-connected-in-fk}
If vertices $u, v \in V$ are connected in forest $F_{k}$, then they are connected in every preceding forest $F_{j}$ for $1 \leq j < k$.
\end{lemma}
\begin{proof}
Fix $j$ with $1 \leq j < k$. Since the forests $E_1,\ldots, E_{|E|}$ are pairwise edge-disjoint and $j < k$, we have $E_k \subseteq E \setminus \left(\bigcup_{i = 1}^{j - 1} E_i\right)$. If $u$ and $v$ are connected in $F_k$, there exists a path $P_{uv} \subseteq E_k$, and by the inclusion above, $P_{uv} \subseteq E \setminus \left(\bigcup_{i = 1}^{j - 1} E_i\right)$. Since $F_j$ is a maximal spanning forest of $E \setminus \left(\bigcup_{i = 1}^{j - 1} E_i\right)$ (Lemma 2.5 in~\cite{Nagamochi92}), $u$ and $v$ are connected in $F_j$.
\end{proof}

The following theorem is the central correctness result for \textsc{HandleAddition} algorithm; its proof combines Lemmas~\ref{lemma:sets-of-connected-components} and~\ref{lemma:two-endpoints-connected-in-fk}.

\begin{theorem}
\label{theorem:edge-cascade-correctness}
Let $G' = \left(V, \bigcup_{i=1}^k E_i \right)$ be the sparse certificate maintained by Algorithm~\ref{algorithm:handling-edge-additions}. If \textsc{HandleAddition}($e_{new}$) discards an edge $e_{old} = \{u, v\}$ ($e_{new} \neq e_{old}$), then $\lambda(u, v; \left(G' \cup \{e_{new}\} \right) \setminus \{e_{old}\}) \geq k$.
\end{theorem}
\begin{proof}
We maintain the following invariant: after every \textsc{TryAdd} operation, each forest $F_j$ ($1 \leq j \leq k$) is a maximal spanning forest of $E \setminus (\bigcup_{i = 1}^{j - 1} E_i)$, and the forests $E_1,\ldots, E_{k}$ remain pairwise edge-disjoint. The invariant holds initially by the Nagamochi–Ibaraki construction (Lemma~\ref{lemma:sparse-cert-lemma}). We first show it is preserved by \textsc{HandleAddition}, then use it to conclude $e_{old}$ redundancy. \par
We will consider two cases. First, let us assume that \textsc{TryAdd}($LC_j$, $x$, $y$) returns null value (edge is absorbed into forest $F_j$). The edge $\{x, y\}$ connects two connected components of $F_j$. Since two connected components of $F_j$ are connected by $\{x, y\}$, any edge from $\bigcup_{i = j + 1}^{k} E_i$ still induce a cycle in $F_j$, so $F_j$ remains maximal. For each $l > j$, the edges are moved between forests with edge sets $E_1, E_2, \ldots, E_j$ -- contained in $\bigcup_{i = 1}^{l - 1} E_i$ -- and the new edge $e_{new}$ is present in $E$ and $E_1$. Thus, the set of available edges $E \setminus \bigcup_{i = 1}^{l - 1} E_i$ remains unchanged, so $F_l$ remains maximal for $l > j$. \par

Now, assume that \textsc{TryAdd}($LC_j$, $x$, $y$) returns non-null value (edge was displaced). The operation replaces one edge in $E_j$ with another, and by Lemma~\ref{lemma:sets-of-connected-components} the connected components of $F_j$ are unchanged. Disjointness is maintained because the displaced edge is removed from $E_j$ before being passed to the next forest. Before the operation, $F_j$ is a maximal spanning forest of $E \setminus \left(\bigcup_{i = 1}^{j - 1} E_i\right)$, so every edge in $\bigcup_{i = j+1}^{k} E_i$ induces a cycle in $F_j$. Since \textsc{TryAdd} preserves connected components of $F_j$ (Lemma~\ref{lemma:sets-of-connected-components}), every such edge still creates a cycle. The displaced edge also has both endpoints in the same component of $F_j$, so it would induce a cycle if re-added. Therefore $F_j$ remains maximal. \par

Let $e_{old} = \{u, w\}$ be the edge discarded from $F_k$. Since $w = $\textsc{ parent}($u$), the vertices $u$ and $w$ were in the same connected component of $F_k$ before the swap, and they remain so afterward by Lemma~\ref{lemma:sets-of-connected-components}. By Lemma~\ref{lemma:two-endpoints-connected-in-fk}, $u$ and $w$ are connected in every $F_j$ for $1 \leq j \leq k$. Each forest $F_j$ therefore contains a path 
$P_j$ from $u$ to $w$ using edges from $E_j$. Since $E_1, E_2, \ldots, E_k$ are pairwise disjoint, the paths $P_1, P_2, \ldots, P_k$ are edge-disjoint. Hence $\lambda(u, w; G') \geq k$. \par
\end{proof}

\subsection{Complexity analysis of sparse certificate maintenance algorithm}
\label{subsection:complexity-of-sparse-certificate-maintenance}

Initialization consists of two steps: computing the forest decomposition $F_1, F_2, \ldots, F_k$ via Algorithm~\ref{algorithm:finding-k-edge-connected-spanning-subgraph} in $O(m)$ time, and building the $k$ Link-Cut structures $LC_1, LC_2, \ldots, LC_k$. The latter requires $O(kn)$ \textsc{MakeTree} and \textsc{link} operations, each costing $O(\log n)$ amortized time, for a total initialization cost of $O(m + k \cdot n \log n)$. \par

For each edge addition, \textsc{HandleAddition} calls \textsc{TryAdd} on at most $k$ forests. Each \textsc{TryAdd} call performs $O(1)$ Link-Cut tree operations (\textsc{FindRoot}, \textsc{cut}, \textsc{link}), so the amortized cost per edge addition is $O(k \log n)$. \par

\section{Restoration after Edge Removal}
\label{section:restoration-after-edge-removal}

\subsection{Edge-connectivity restoration algorithm}
\label{subsection:connectivity-restoration-algorithm}

This section addresses restoring $k$-edge-connectivity after an edge deletion. Let $e = \{u, v\}$ be the deleted edge and let $G' = G \setminus \{e\}$. If $\lambda(G') \geq k$, no action is needed. Otherwise, the algorithm must compute a set of augmenting edges $E_{aug}$ such that $\lambda(G' \cup E_{aug}) \geq k$. We compute a maximum $u$-$v$ flow in $G'$ using Dinic's algorithm~\cite{Dinitz70}, whose complexity on unit-capacity networks is $O(\min(n^{2/3}, m^{1/2}) \cdot m)$ by Even and Tarjan~\cite{Shimon75}, and use the resulting residual graph to identify $E_{aug}$. \par

We model each undirected edge $\{u, v\} \in E$ as a pair of directed arcs $(u, v)$ and $(v, u)$, each with unit capacity; for background on flow networks and residual graphs we refer to~\cite{Cormen2009}. In the resulting residual graph $G_f$, a saturated arc $(u, v)$ with $f(u, v) = 1$ is replaced by a backward arc $(v, u)$ with residual capacity $2$, while an unsaturated arc retains both directions with residual capacity $1$. \par

The procedure \textsc{HandleDeletion} (Algorithm~\ref{algorithm:handling-edge-deletions}) works as follows. Given the deleted edge $e = \{u, v\}$, it computes a maximum $u$-$v$ flow on $G' = G \setminus \{e\}$. If the flow value is still at least $k$, the edge was redundant and no action is needed. Otherwise, let $G_f$ denote the residual graph. The algorithm computes $S$, the set of vertices reachable from $u$ in $G_f$, and $T$, the set of vertices from which $v$ is reachable in $G_f$; both sets are obtained by BFS or DFS algorithm on $G_f$. \par

\begin{algorithm}[p]
\caption{Procedure for handling edge deletion operations.}
\label{algorithm:handling-edge-deletions}
\begin{algorithmic}[1]

\Procedure{HandleDeletion}{$e$}
    \State $(u, v) \gets \text{endpoints}(e)$ \label{line:retrieving-endpoints}
    \State $G \gets G \setminus \{e\}$
    \State $flow \gets \Call{DinicMaxFlow}{G, u, v}$ \label{line:dynic-max-flow-execution}
    \item[]
    \If{$flow \geq k$} \label{line:nothing-to-do}
        \State \Return
    \ElsIf{$flow = k - 1$}
        \State $G_{f} \gets$ Residual Graph from Dinic's execution on $G$ \label{line:get-residual-graph}
        \State $S \gets \{ w \in V \mid w \text{ is reachable from } u \text{ in } G_{f} \}$
        \State $T \gets \{ w \in V \mid v \text{ is reachable from } w \text{ in } G_{f} \}$ \label{line:compute-set-t}
        \item[]
        \If{$S = \{u\} \wedge T = \{v\}$}
            \State $u' \gets \text{arbitrarily chosen vertex in } G$ \label{line:singleton-get-intermediate-vertex}
            \State $G \gets G \cup \{\{u, u'\}, \{u', v\}\}$ \label{line:add-two-augmenting-edges}
        \Else
            \State $u' \gets \text{arbitrarily chosen vertex from } S$ \label{line:first-bridging-endpoint}
            \State $v' \gets \text{arbitrarily chosen vertex from } T$
            \State $G \gets G \cup \{u', v'\}$ \label{line:add-bridging-edge}
        \EndIf
    \EndIf
\EndProcedure

\end{algorithmic}
\end{algorithm}

\begin{figure}[p]
  \centering
  \begin{subfigure}[b]{0.3\textwidth}
    \centering
    \begin{tikzpicture}
      \node[draw, circle, inner sep=2pt, minimum size=20pt] (1) at (0, 0) {$v_1$};
      \node[draw, circle, inner sep=2pt, minimum size=20pt] (3) at (1.5, 1) {$v_3$};
      \node[draw, circle, inner sep=2pt, minimum size=20pt] (4) at (1.5, -1) {$v_4$};
      \node[draw, circle, inner sep=2pt, minimum size=20pt] (2) at (3, 0) {$v_2$};
     
      \draw[thick] (1) -- (3);
      \draw[thick] (3) -- (2);
      \draw[thick] (3) -- (4);
      \draw[thick] (4) -- (2);
    \end{tikzpicture}
    \caption{$G' = G \setminus \{v_1, v_2\}$}
  \end{subfigure}
  \hfill
  \begin{subfigure}[b]{0.3\textwidth}
    \centering
    \begin{tikzpicture}
     
      \node[draw, circle, inner sep=2pt, minimum size=20pt, fill=blue!20] (1) at (0, 0) {$v_1$};
     
      \node[draw, circle, inner sep=2pt, minimum size=20pt, fill=red!20] (3) at (1.5, 1) {$v_3$};
      \node[draw, circle, inner sep=2pt, minimum size=20pt, fill=red!20] (4) at (1.5, -1) {$v_4$};
      \node[draw, circle, inner sep=2pt, minimum size=20pt, fill=red!20] (2) at (3, 0) {$v_2$};
   
      \draw[->, >=stealth, thick] (3) -- (1);
      \draw[->, >=stealth, thick] (2) -- (3);
     
      \draw[->, >=stealth, thick] (3) to[bend left=12] (4);
      \draw[->, >=stealth, thick] (4) to[bend left=12] (3);
      \draw[->, >=stealth, thick] (4) to[bend left=12] (2);
      \draw[->, >=stealth, thick] (2) to[bend left=12] (4);
    
      \node[blue, font=\footnotesize] at (-0.7, 0) {$S$};
      \node[red, font=\footnotesize] at (3.7, 0) {$T$};
    \end{tikzpicture}
    \caption{Residual graph $G_f$}
  \end{subfigure}
  \hfill
  \begin{subfigure}[b]{0.3\textwidth}
    \centering
    \begin{tikzpicture}
      \node[draw, circle, inner sep=2pt, minimum size=20pt] (1) at (0, 0) {$v_1$};
      \node[draw, circle, inner sep=2pt, minimum size=20pt] (3) at (1.5, 1) {$v_3$};
      \node[draw, circle, inner sep=2pt, minimum size=20pt] (4) at (1.5, -1) {$v_4$};
      \node[draw, circle, inner sep=2pt, minimum size=20pt] (2) at (3, 0) {$v_2$};
      
      \draw[thick] (1) -- (3);
      \draw[thick] (3) -- (2);
      \draw[thick] (3) -- (4);
      \draw[thick] (4) -- (2);
     
      \draw[red, dashed, thick] (1) -- (4);
    \end{tikzpicture}
    \caption{Augmented graph $G''$}
  \end{subfigure}
  \caption{\textsc{HandleDeletion}($\{v_1, v_2\}$): single augmenting edge ($|T| \geq 2$).}
  \label{fig:one-edge-case}
\end{figure}

\begin{figure}[p]
  \centering
 
  \begin{subfigure}[b]{0.3\textwidth}
    \centering
    \begin{tikzpicture}
        \foreach \i/\v in {1/v_1, 2/v_2, 3/v_3, 4/v_4, 5/v_5, 6/v_6} {
            \node[draw, circle, inner sep=2pt, minimum size=20pt] (\i) at ({180 - (\i-1)*60}:1.5cm) {$\v$};
        }
        \foreach \a/\b in {2/3, 3/4, 4/5, 5/6, 6/1} {
            \draw (\a) -- (\b);
        }
    \end{tikzpicture}
    \caption{$G' = G \setminus \{v_1, v_2\}$}
  \end{subfigure}
  \hfill
  \begin{subfigure}[b]{0.3\textwidth}
    \centering
    \begin{tikzpicture}
        
        \node[draw, circle, inner sep=2pt, minimum size=20pt, fill=blue!20] (1) at ({180 - 0*60}:1.5cm) {$v_1$};
        
         \node[draw, circle, inner sep=2pt, minimum size=20pt, fill=red!20] (2) at ({180 - 1*60}:1.5cm) {$v_2$};
       
         \foreach \i/\v in {3/v_3, 4/v_4, 5/v_5, 6/v_6} {
            \node[draw, circle, inner sep=2pt, minimum size=20pt] (\i) at ({180 - (\i-1)*60}:1.5cm) {$\v$};
         }
        
        \draw[->, >=stealth, thick] (2) -- (3);
        \draw[->, >=stealth, thick] (3) -- (4);
        \draw[->, >=stealth, thick] (4) -- (5);
        \draw[->, >=stealth, thick] (5) -- (6);
        \draw[->, >=stealth, thick] (6) -- (1);
        Labels for S and T
        \node[blue, font=\footnotesize] at (180:2.7cm) {$S = \{v_1\}$};
        \node[red, font=\footnotesize] at (105:2.0cm) {$T = \{v_2\}$};
    \end{tikzpicture}
    \caption{Residual graph $G_f$}
  \end{subfigure}
  \hfill
  \begin{subfigure}[b]{0.3\textwidth}
    \centering
    \begin{tikzpicture}
        
        \foreach \i/\v in {1/v_1, 2/v_2, 3/v_3, 4/v_4, 5/v_5, 6/v_6} {
            \node[draw, circle, inner sep=2pt, minimum size=20pt] (\i) at ({180 - (\i-1)*60}:1.5cm) {$\v$};
        }
       
        \draw[thick] (2) -- (3);
        \draw[thick] (3) -- (4);
        \draw[thick] (4) -- (5);
        \draw[thick] (5) -- (6);
        \draw[thick] (6) -- (1);
       
        \draw[red, dashed, thick] (1) -- (4);
        \draw[red, dashed, thick] (4) -- (2);
    \end{tikzpicture}
    \caption{Augmented graph $G''$}
  \end{subfigure}
  \caption{\textsc{HandleDeletion}($\{v_1, v_2\}$): two augmenting edges ($|S| = |T| = 1$).}
  \label{fig:two-edge-case}
\end{figure}

The algorithm distinguishes two cases. If $|S| \geq 2$ or $|T| \geq 2$, there exists a pair $u' \in S$, $v' \in T$ with $\{u', v'\} \neq \{u, v\}$; the algorithm adds $\{u', v'\}$ to $G'$ (Algorithm~\ref{algorithm:handling-edge-deletions}, lines~\ref{line:first-bridging-endpoint}-\ref{line:add-bridging-edge}). If $S = \{u\}$ and $T = \{v\}$, the only edge from $S$ to $T$ is the deleted edge $\{u, v\}$, so instead the algorithm selects an arbitrary vertex $w \in V \setminus \{u, v\}$ and adds the two edges $\{u, w\}$ and $\{w, v\}$ (Algorithm~\ref{algorithm:handling-edge-deletions}, lines~\ref{line:singleton-get-intermediate-vertex}-\ref{line:add-two-augmenting-edges}). \par

\subsection{Correctness and complexity of edge-connectivity restoration algorithm}

We show that edges added by \textsc{HandleDeletion} increases the maximum $u$-$v$ flow by one (Lemma~\ref{lemma:disjointness-and-augmentation}), and that the full procedure restores $\lambda(G') \geq k$ (Theorem~\ref{theorem:algorithm-correctness}). \par

\begin{lemma}
\label{lemma:disjointness-and-augmentation}
  Let $f$ be a maximum $u$-$v$ flow in $G' = G \setminus \{e\}$ ($e = \{u, v\}$), let $G_f$ be the residual graph, and define $S = \{w \in V : w \text{ is reachable from } u \text{ in } G_f\}$ and $T = \{w \in V : v \text{ is reachable from } w \text{ in } G_f\}$. Then:
\begin{enumerate}
    \item $S \cap T = \emptyset$.
    \item For any $u' \in S$ and $v' \in T$, adding the undirected edge $\{u', v'\}$ to $G'$ increases the maximum $u$-$v$ flow by one.
\end{enumerate}
\end{lemma}

\begin{proof}
\begin{enumerate}
\item If $w \in S \cap T$, then $G_f$ contains a path from $u$ to $w$ (since $w \in S$) and a path from $w$ to $v$ (since $w \in T$). Concatenating these paths yields an augmenting $u$-$v$ path, contradicting maximality of $f$.
\item Adding edge $\{u', v'\}$ to $G'$ introduces the arc $(u', v')$ in the residual graph. Since $u' \in S$, there is a path from $u$ to $u'$ in $G_f$; since $v' \in T$, there is a path from $v'$ to $v$ in $G_f$. Concatenating these with the new arc gives an augmenting path from $u$ to $v$, increasing the flow by one.
\end{enumerate}
\end{proof}

\begin{theorem}
\label{theorem:algorithm-correctness}
Let $G' = G \setminus \{e\}$ where $e = \{u, v\}$ and $\lambda(G') < k$. Let $G''$ be the augmented graph produced by Algorithm~\ref{algorithm:handling-edge-deletions}. Then $\lambda(G'') \geq k$, and $|E(G'') \setminus E(G')| \leq 2$.
\end{theorem}
\begin{proof}
By Lemma~\ref{lemma:cut-localization-lemma}, $\lambda(u, v; G') = k - 1$, and every cut of size less than $k$ in $G'$ separates $u$ and $v$. The max $u$-$v$ flow in $G'$ therefore equals $k - 1$.

\begin{enumerate}
    \item Case 1: $|S| \geq 2$ or $|T| \geq 2$. There exist $u' \in S$ and $v' \in T$ with $\{u', v'\} \neq \{u, v\}$. The algorithm adds $\{u', v'\}$. By Lemma~\ref{lemma:disjointness-and-augmentation}, this increases the maximum $u$-$v$ flow to $k$. Since every cut of size less than $k$ separates $u$ and $v$ (Lemma~\ref{lemma:cut-localization-lemma}), and the minimum $u$-$v$ cut now has size $k$, we have $\lambda(G'') \geq k$. One edge is added.
    \item Case 2: $S = \{u\}$ and $T = \{v\}$. The algorithm selects $w \in V \setminus \{u, v\}$ and adds $\{u, w\}$ and $\{w, v\}$. This creates the augmenting path $u \rightarrow w \rightarrow v$, increasing the maximum $u$-$v$ flow to $k$. The same argument as Case 1 gives $\lambda(G'') \geq k$. Two edges are added.
\end{enumerate}

In both cases $|E(G'') \setminus E(G')| \leq 2$ holds.
\end{proof}

The cost of Algorithm~\ref{algorithm:handling-edge-deletions} is dominated by one execution of Dinic's algorithm on the sparsified graph $G'$. After applying the Nagamochi–Ibaraki sparsification~\cite{Nagamochi92}, the graph has $m = O(kn)$ edges while preserving all cuts of size at most $k$. By the Even–Tarjan bound~\cite{Shimon75}, Dinic's algorithm on a unit-capacity network runs in $O(\min(n^{2/3}, m^{1/2}) \cdot m)$ time. Substituting $m = O(kn)$, this yields $O(n^{2/3} \cdot kn) = O(kn^{5/3})$ and $O((kn)^{1/2} \cdot kn) = O(k^{3/2} \cdot n^{3/2})$, respectively. Since $k$ is a fixed constant, $k^{1/2} \cdot n^{1/2} < n^{2/3}$ for sufficiently large $n$, so the second branch dominates and the total cost is $O(k^{3/2} \cdot n^{3/2})$. The subsequent BFS/DFS algorithm to compute $S$ and $T$ runs in $O(kn)$ time and does not affect the asymptotic bound. \par

\bibliographystyle{plain}
\bibliography{bibliography} 

@article{Nagamochi92,
author = {Nagamochi, Hiroshi and Ibaraki, Toshihide},
title = {A linear-time algorithm for finding a sparsek-connected spanning subgraph of ak-connected graph},
year = {1992},
issue_date = {Jun 1992},
publisher = {Springer-Verlag},
address = {Berlin, Heidelberg},
volume = {7},
number = {1–6},
issn = {0178-4617},
url = {https://doi.org/10.1007/BF01758778},
doi = {10.1007/BF01758778},
pages = {583–596},
numpages = {14}
}

@inproceedings{Tarjan81,
author = {Sleator, Daniel D. and Tarjan, Robert Endre},
title = {A data structure for dynamic trees},
year = {1981},
isbn = {9781450373920},
publisher = {Association for Computing Machinery},
address = {New York, NY, USA},
url = {https://doi.org/10.1145/800076.802464},
doi = {10.1145/800076.802464},
booktitle = {Proceedings of the Thirteenth Annual ACM Symposium on Theory of Computing},
pages = {114–122},
numpages = {9},
location = {Milwaukee, Wisconsin, USA},
series = {STOC '81}
}

@article{Shimon75,
author = {Even, Shimon and Tarjan, R. Endre},
title = {Network Flow and Testing Graph Connectivity},
journal = {SIAM Journal on Computing},
volume = {4},
number = {4},
pages = {507-518},
year = {1975},
doi = {10.1137/0204043},
URL = {https://doi.org/10.1137/0204043},
eprint = {https://doi.org/10.1137/0204043}
}

@article{Tarjan85,
author = {Sleator, Daniel Dominic and Tarjan, Robert Endre},
title = {Self-adjusting binary search trees},
year = {1985},
issue_date = {July 1985},
publisher = {Association for Computing Machinery},
address = {New York, NY, USA},
volume = {32},
number = {3},
issn = {0004-5411},
url = {https://doi.org/10.1145/3828.3835},
doi = {10.1145/3828.3835},
journal = {J. ACM},
month = jul,
pages = {652–686},
numpages = {35}
}

@article{Dinitz70,
author = {Dinitz, Yefim},
year = {1970},
month = {01},
pages = {1277-1280},
title = {Algorithm for Solution of a Problem of Maximum Flow in Networks with Power Estimation},
volume = {11},
journal = {Soviet Math. Dokl.}
}

@book{Cormen2009,
  title={Introduction to Algorithms},
  author={Cormen, Thomas H. and Leiserson, Charles E. and Rivest, Ronald L. and Stein, Clifford},
  publisher={The MIT Press},
  year={2009},
  edition={3rd},
  isbn={978-0-262-03384-8} 
}

@article{Thorup2007,
  title={Fully dynamic minimum cut},
  author={Thorup, Mikkel},
  journal={Combinatorica},
  volume={27},
  number={1},
  pages={91--127},
  year={2007},
  publisher={Springer}
}

@book{Diestel2017,
author = {Diestel, Reinhard},
title = {Graph Theory},
year = {2017},
isbn = {3662536218},
publisher = {Springer Publishing Company, Incorporated},
edition = {5th},
pages = {59--67}
}

@article{Gomory61,
author = {Gomory, R. E. and Hu, T. C.},
title = {Multi-Terminal Network Flows},
journal = {Journal of the Society for Industrial and Applied Mathematics},
volume = {9},
number = {4},
pages = {551-570},
year = {1961},
doi = {10.1137/0109047},
URL = {https://doi.org/10.1137/0109047},
eprint = {https://doi.org/10.1137/0109047}
}

@article{Karger2000,
author = {Karger, David R.},
title = {Minimum cuts in near-linear time},
year = {2000},
issue_date = {Jan. 2000},
publisher = {Association for Computing Machinery},
address = {New York, NY, USA},
volume = {47},
number = {1},
issn = {0004-5411},
url = {https://doi.org/10.1145/331605.331608},
doi = {10.1145/331605.331608},
journal = {J. ACM},
month = jan,
pages = {46–76},
numpages = {31},
}

@article{Henzinger2020,
author = {Henzinger, Monika and Rao, Satish and Wang, Di},
title = {Local Flow Partitioning for Faster Edge Connectivity},
journal = {SIAM Journal on Computing},
volume = {49},
number = {1},
pages = {1-36},
year = {2020},
doi = {10.1137/18M1180335},
URL = {https://doi.org/10.1137/18M1180335},
eprint = {https://doi.org/10.1137/18M1180335}
}

@misc{Li2021,
      title={Deterministic Mincut in Almost-Linear Time}, 
      author={Jason Li},
      year={2021},
      eprint={2106.05513},
      archivePrefix={arXiv},
      primaryClass={cs.DS},
      url={https://arxiv.org/abs/2106.05513}, 
}

@article{Goranci2018,
author = {Goranci, Gramoz and Henzinger, Monika and Thorup, Mikkel},
title = {Incremental Exact Min-Cut in Polylogarithmic Amortized Update Time},
year = {2018},
issue_date = {April 2018},
publisher = {Association for Computing Machinery},
address = {New York, NY, USA},
volume = {14},
number = {2},
issn = {1549-6325},
url = {https://doi.org/10.1145/3174803},
doi = {10.1145/3174803},
journal = {ACM Trans. Algorithms},
month = mar,
articleno = {17},
numpages = {21}
}

@misc{Demaine2012,
  author       = {Erik Demaine},
  title        = {6.851: Advanced Data Structures, Lecture 19},
  year         = {2012},
  month        = {May},
  note         = {Spring 2012. Scribes: Justin Holmgren, Jing Jian, Maksim Stepanenko (2012), and Mashhood Ishaque (2007).}
}

@article{Santi2005,
author = {Santi, Paolo},
title = {Topology control in wireless ad hoc and sensor networks},
year = {2005},
issue_date = {June 2005},
publisher = {Association for Computing Machinery},
address = {New York, NY, USA},
volume = {37},
number = {2},
issn = {0360-0300},
url = {https://doi.org/10.1145/1089733.1089736},
doi = {10.1145/1089733.1089736},
journal = {ACM Comput. Surv.},
month = jun,
pages = {164–194},
numpages = {31}
}

@article{Grotschel92,
 ISSN = {0030364X, 15265463},
 URL = {http://www.jstor.org/stable/171455},
 author = {Martin Grötschel and Clyde L. Monma and Mechthild Stoer},
 journal = {Operations Research},
 number = {2},
 pages = {309--330},
 publisher = {INFORMS},
 title = {Computational Results with a Cutting Plane Algorithm for Designing Communication Networks with Low-Connectivity Constraints},
 urldate = {2026-03-03},
 volume = {40},
 year = {1992}
}

\end{document}